%% file: pra2.tex
\documentclass{llncs}
\usepackage{times}
\usepackage{amsfonts,  amsmath, amssymb}

\usepackage[compress]{natbib}
\bibpunct{[}{]}{,}{n}{}{,}

\usepackage[dvips]{color}
\usepackage{graphicx}
\usepackage{epstopdf}

\newcommand{\xvec}{x}
\newcommand{\zvec}{z}

\newcommand{\std}{ \sigma }

\newcommand{\var}{ v }
\newcommand{\kappavar}{\kappa}

\newcommand{\route}{p}
\newcommand{\paths}{\mathcal{P}}

\newcommand{\flow}{f}

\newcommand{\N}{\mathbb{N}}
\newcommand{\R}{\mathbb{R}}

\newcommand{\pathcost}{Q}

\newcommand{\pcdev}{\pathcost^{\gamma}}
\newcommand{\pc}{\pathcost^0}
\newcommand{\prav}{PRA} 
\newcommand{\pras}{PRA} 

\newcommand{\xx}{x}   
\newcommand{\zz}{z}   
\newcommand{\CC}{C} 
\newcommand{\GV}{V}
\newcommand{\GA}{E}

\definecolor{myred}{rgb}{0.95, 0.01, 0.01}
\definecolor{myblue}{rgb}{0.01, 0.01, 0.9}

\newcommand{\todo}[1]{}

\title{Asymptotically tight bounds for inefficiency\\in risk-averse selfish routing}
\author{Thanasis Lianeas\inst{1} \and Evdokia Nikolova\inst{1} \and Nicolas E. Stier-Moses\inst{2}}
\institute{University of Texas at Austin \and Universidad Torcuato Di Tella}

\begin{document}

\maketitle

\begin{abstract}
We consider a nonatomic selfish routing model with independent stochastic travel times, represented by mean and variance latency functions for each edge that depend on their flows. In an effort to decouple the effect of risk-averse player preferences from selfish behavior on the degradation of system performance, \citet{Nikolova-Stier15} defined the concept of the {\em price of risk aversion} as the worst-case ratio of the cost of an equilibrium with risk-averse players (who seek risk-minimizing paths, for an appropriate definition of risk) and that of an equilibrium with risk-neutral users (who minimize the mean travel time of a path). For risk-averse users who seek to minimize the mean plus variance of travel time on a path, they proved an upper bound on the price of risk aversion, which is independent of the mean and variance latency functions, and grows linearly with the size of the graph and players' risk-aversion. 

In this follow-up paper, we provide a matching lower bound for graphs with number of vertices equal to powers of two, via the construction of a graph family inductively generated from the Braess graph.  In contrast to these {\em topological} bounds that depend on the topology of the network, we also provide conceptually different bounds, which we call {\em functional}. These bounds depend on the class of mean latency functions that are allowed and provide characterizations that are independent of the network topology.  The functional upper bound was first derived by \citet{meir-wronggame} in a different context with different techniques; we offer a simpler, direct proof that is inspired by a classic proof technique using variational inequalities~\cite{css-congestion}.  We also supplement the upper bound with a new asymptotically-tight lower bound, derived from the same graph construction as the topological lower bound.  Thus, we offer a conceptually new perspective and understanding of both this and classic congestion game settings in terms of the {\em functional} versus {\em topological} view of efficiency loss.  

Our third contribution is a tight bound on the price of risk aversion for a family of graphs that generalize series-parallel graphs and the Braess graph.  That bound applies to both users minimizing the mean plus variance (mean-var) of a path, as well as to users minimizing the mean plus standard deviation (mean-stdev) of a path---a much more complex model of risk-aversion due to the cost of a path being non-additive over edge costs. This is a refinement of previous results in~\cite{Nikolova-Stier15} that characterized the price of risk-aversion for series-parallel graphs and for the Braess graph.  The main question left open is to upper bound the price of risk-aversion in the mean-stdev model for general graphs; our lower bounds apply to both the mean-var and the mean-stdev models. 
\end{abstract}


\input{intro.tex}
\input{model.tex}

\input{structlowerbound.tex}

\input{functionalbounds.tex}

\input{structupperbound.tex}
\input{conclusion.tex}


\bibliographystyle{abbrvnat}
\small
\bibliography{Bibfiles/rwe,Bibfiles/soue,Bibfiles/routeguid,Bibfiles/telecom,Bibfiles/eddie-thesis,Bibfiles/robust,Bibfiles/risk,Bibfiles/risk2}

\end{document}

%% file: intro.tex
\section{Introduction}\label{sec:intro}

One of the key challenges of making optimal routing decisions is the phenomenon
of congestion: the fact that the travel time along a link increases with the
number of users on that link.  Thus, a user deciding on her optimal route needs
to take into account the routing decisions of other users in the network. 
Networks subject to congestion lead to significant tensions between the local
goals of users to minimize their travel times and the global goal of the
network planner to minimize the total travel time of all users.  The challenge
was investigated in the early work of \citet{wardrop-traffic} and
\citet{beckmann-transportation} by formalizing congestion effects
into a game theoretic model of routing and by defining and analyzing the
traffic assignments or flows resulting from the two conflicting goals, known as
the Wardrop equilibrium and the social optimum, respectively. 

The desire to understand and precisely quantify the severity of the tension
between equilibrium and social optimum, or, in other words, to quantify the
degradation of system performance due to selfish behavior, inspired the
definition of the price of anarchy, which, informally speaking, measures
equilibrium inefficiency relative to a socially-optimal solution.
Consequently, routing games have been central to the development of algorithmic
game theory and have inspired the intensive study of the price of anarchy in
many settings with incentives beyond routing.  At the same time, routing games
continue to be a rich source of new research questions driven by the need to
add realism to network models and to improve real-life applications. 

Indeed, routing is fundamental to diverse applications affecting everyday life
including transportation, telecommunication networks, robotics, task planning,
etc.  All of these applications suffer inherent uncertainty in the network
parameters such as travel times and demands, which can significantly alter
individual routing choices and completely throw off a predicted equilibrium
solution and its efficiency, due to players' risk aversion. For
example, \citet{PiliourasNS:2013} illustrate that the price of
anarchy results are extremely sensitive to the modeling of risk-averse
preferences and show that the price of anarchy may in some models decrease
while in others become unbounded.  

Incorporating risk-aversion in routing games is particularly challenging in
general due to the often nonlinear nature of risk attitudes.  For example, even
finding a best response, which is a minimum-risk path according to some
appropriate definition of risk, may lead to an algorithmic problem of unknown
complexity that we currently do not know how to solve in polynomial
time~\cite{brand-shochSP,nikolova-ssp,nikolova-approx}.  On the other hand,
even for simpler risk-averse objectives that are additive and algorithmically
tractable, understanding the effect of risk on equilibrium inefficiency may
require fundamentally different techniques from the ones used so far to analyze
the price of anarchy~\cite{Nikolova-Stier15}.

There has been an increased effort in recent years to model risk-averse
preferences in routing games and understand the effect of such player
preferences on network
equilibria~\cite{os-rweTS,NikolovaM11,nie-perceq,AngelidakisFL13,PiliourasNS:2013,nikolova-stochWE,Nikolova-Stier15,cominetti-riskaverserouting}.
We follow the mean-variance risk model considered by 
\citet{Nikolova-Stier15} as well as the mean-stdev model considered
there and in their previous work~\cite{nikolova-stochWE}.
Risk-averse agents are postulated to minimize a linear combination of the mean
and variance of a path, or the mean and standard deviation of a path,
respectively.  We defer the reader to this earlier work for the motivation and
criticisms of these models of risk-aversion in network settings. 

In an effort to decouple the effect of risk attitudes from the effect of
selfish behavior on the degradation of system performance, 
\citet{Nikolova-Stier15} defined the concept of {\em price of risk aversion
($\prav$)} as the worst-case ratio of the cost of a risk-averse equilibrium to that
of a risk-neutral equilibrium (namely, the equilibria when agents are
risk-averse and risk-neutral, respectively).  The main result in their
paper was an upper bound on the price of risk aversion in general graphs that
is linear in the number of vertices of the graph.  Specifically, the bound was
shown to be $1+\kappa\gamma n/2$, where $\kappa$ is the worst-case
variance-to-mean ratio of an edge at equilibrium, $\gamma$ is the coefficient
of risk-aversion and $n$ is the number of vertices in the graph. 

\paragraph{Our contribution}

In this follow-up paper to \citet{Nikolova-Stier15}, we
provide a tight lower bound to the price of risk aversion, matching $1+\kappa\gamma
n/2$, through the construction of a graph family with number of
vertices that are powers of two
(Theorem~\ref{cor:altertight}).  Whereas the upper bound on the price of
risk-aversion was established for general graphs only in the mean-variance
model and was left open for the mean-stdev model, the lower bound provided here
applies to both models.  The upper bound was based on establishing the existence of
an alternating $s$-$t$-path consisting of forward and backward edges for which equilibria satisfy a certain property,
and seeing that the alternating path can have at most $n/2$ alternations.
Constructing the worst-case family presented in this paper involves finding an
instance in which the alternating path goes through every vertex in the graph and
alternates between forward and backward edges at every internal vertex. We
achieve this by inductively defining a graph family with appropriate 
mean and variance functions for each edge, using the topology of the Braess graph.

A key feature of these upper and lower bounds to the price of
risk-aversion is that they are independent of latency functions, though highly
dependent on the topology of the graph.  We call these results {\em
topological}, as opposed to the {\em functional} nature of existing price of
anarchy results, which quantify the price of anarchy in terms of the class of
allowed latency functions and which provide bounds independent of the network
topology~\cite{roughgarden-priceanarchy,css-capsoue}. For example, the famous
price of anarchy bound of $4/3$ holds for linear latency functions
and arbitrary graph topologies,
and it is {\em unbounded} for arbitrary latency functions.  In contrast, the
price of risk aversion upper bound of \citet{Nikolova-Stier15}, as well as our
lower bounds, depend on the network topology and are independent of the latency
functions. 

Our second contribution bridges this topological vs.\ functional view of
equilibrium inefficiency, by developing an asymptotically-tight {\em
functional} bound for the mean-variance model. This new bound depends on the
class of allowed latency functions and is independent of the network topology,
as with the classic price of anarchy results. In particular, using a
variational inequality characterization of equilibria proposed by
\citet{css-capsoue}, we show that the price of risk aversion is upper bounded
by $(1+\kappa\gamma)(1-\mu)^{-1}$ for $(1,\mu)$-smooth latency functions
(Theorem~\ref{thm:functional-upper}).  This implies, for example, that the
price of risk aversion is at most $(1+\kappa\gamma)4/3$ for linear latency
functions.  The upper bound can be thought as a generalization of the classic
result of~\cite{roughgarden-priceanarchy} that established that the price of
anarchy equals $(1-\mu)^{-1}$ because when there is no variability $\kappa=0$.
Furthermore, we show that our bound is asymptotically tight by providing a
matching lower bound (Theorem~\ref{thm:functional-lower}).  We note that the
upper bound was proved using different techniques and in a slightly different
context, by \citet{meir-wronggame}; we provide more details in the related work
section below.  Finally, we remark that for unrestricted functions, the
functional upper bounds become vacuous since $\mu=1$, which provides further
support for the topological analysis of Section~\ref{sec:lb}.  

Finally, our third contribution provides a tight bound on the price of risk
aversion under the mean-stdev model for a family of graphs that generalizes
series-parallel graphs (specifically, the family of graphs where
the domino-with-ears graph, shown in Fig.~\ref{fig:forbiddenGraphs}(b),
is a forbidden minor).
The mean-stdev model is significantly more difficult to analyze than the
mean-var model due to its non-additive nature, namely the mean plus standard
deviation of a path cannot be decomposed as a sum of costs over the edges in
the path. \citet{Nikolova-Stier15} proved that the
price of risk aversion in this model is $1+\gamma\kappa$ for series-parallel
graphs, or equivalently the family of instances where the Braess graph is a forbidden minor.
The proof is based on establishing that
this family of graphs admits an alternating path with zero alternations, namely
an alternating path with forward edges only.  En route to establishing a more
general bound, in this paper we extend the analysis to a larger family of
graphs that generalizes series-parallel graphs, and show that this family admits
alternating paths with one alternation and thus has a price of risk aversion of
$1+2\kappa\gamma$.  We remark that this bound applies to the mean-var model
as well, and it refines our understanding on the topology of graphs for which
the price of risk aversion is $1+2\kappa\gamma$, as opposed to the cruder
bound in terms of number of vertices only. 

An intriguing conjecture that we leave open is to show a bound for the
mean-stdev model for general graphs that is equivalent to the corresponding
bound for the mean-var model, namely that the price of risk aversion is at most
$1+\gamma\kappa\eta$, where $\eta$ is the number of forward subpaths in an
alternating path and $\kappa$ is the maximum coefficient of variation of edge
latencies at the equilibrium flow.  Intuitively, the mean plus standard
deviation along a path is upper bounded by the corresponding mean plus
variance.  It is appealing to think that, as a result, the cost of the
mean-stdev risk-averse equilibrium should be upper bounded by that of the
mean-var risk-averse equilibrium.  A corollary of such a bound, combined with our
first contribution here, would be an asymptotically-tight bound on the price of
risk aversion equal to $1+\gamma\kappa n/2$ for the mean-stdev model in general
graphs. 

\paragraph{Related Work}

The most closely related work to ours is that of \citet{Nikolova-Stier15}, who
define the concept of price of risk aversion as the worst-case ratio of the
cost of a risk-averse equilibrium to that of a risk-neutral equilibrium.  For
users that seek to minimize the mean plus $\gamma$ times the variance of a
path, where $\gamma$ is a given constant that parameterizes the degree of
risk-aversion, they show that the price of risk-aversion in general networks is
upper bounded by $1+\kappa\gamma\eta$, where $\eta$ is a parameter that depends
on the graph topology and $\kappa$ is the maximum variance-to-mean ratio. The
remarkable feature of that bound is that it is independent of the mean and
variance latency functions.  The proof is based on establishing the existence
of an alternating $s$-$t$ path of forward and backward edges, in which the
forward edges carry more risk-neutral flow and the backward edges carry more
risk-averse flow.  A question left open by them is whether this bound is tight,
which is what we prove here.  In addition, for both the mean-variance and the
mean-stdev model, Nikolova and Stier-Moses proved a tight bound of the price of
risk-aversion of $1+\kappa\gamma$ for series-parallel graphs.  Here, we extend
this characterization to a tight bound of $1+2\kappa\gamma$ for a wider family
of graphs, namely those where the domino-with-ears graph is a forbidden minor.
Finally, in contrast to this earlier work, which only provides results
depending on the graph topology, here we additionally offer functional bounds
that are independent of the network topology and instead are parametrized by
the class of allowed latency functions.  

The asymptotically-tight functional bounds we present here were inspired by the
recent work of \citet{meir-wronggame}.  In their paper, they prove a result
that compares an equilibrium when players consider a modified cost function to
the social optimum of the original game. As a corollary, they indirectly derive
an upper bound on the price of risk aversion of $(1+\kappa\gamma)(1-\mu)^{-1}$
when cost functions are $(1,\mu)$ smooth.  As we establish in this paper, this
upper bound and that of \citet{Nikolova-Stier15} are of a different type, i.e.,
functional vs.\ topological, which is why they cannot be compared directly.
Our proof of the upper bound relies on a simpler approach that is a
straightforward generalization of the earlier price of anarchy proof based on
variational inequalities put forward by~\cite{css-congestion}. Consequently,
the method allows for an easier comparison and consistency with the traditional
price of anarchy proofs.  We also provide an asymptotically-matching functional
lower bound, which follows from the same graph construction as our topological
lower bound. 

Finally, we mention again that this paper is part of a relatively new and
growing literature exploring the effect of risk aversion on network equilibria
in routing
games~\cite{os-rweTS,NikolovaM11,nie-perceq,AngelidakisFL13,PiliourasNS:2013,nikolova-stochWE,Nikolova-Stier15,cominetti-riskaverserouting}.
We refer the reader to the recent paper by Nikolova and
Stier-Moses~\cite{Nikolova-Stier15} for a more comprehensive review of
additional related work, as well as a detailed discussion on the pros and cons
of the risk-averse models considered here.  We also refer the reader to the
recent survey by \citet{Cominetti15-routingUncert} for a more extensive
review of equilibrium routing under uncertainty.


%% file: model.tex
\section{The Model}\label{sec:model}

We consider a directed graph $G=(\GV,\GA)$ with a single source-sink pair
$(s,t)$ and an aggregate demand of $d$ units of flow that need to be routed
from $s$ to $t$. We let $\paths$ be the set of all feasible paths between
$s$ and $t$. We encode {the decisions of the symmetric} players as a flow vector $\flow
=(\flow_{\route})_{\route\in\paths}\in \R^{|\paths|}_+$ over all paths. Such
a flow is feasible when demand is satisfied, as given by the constraint
$\sum_{\route\in\paths} \flow_{\route}=d$. For notational simplicity, we
denote the flow on a directed edge $e$ by $\flow_e = \sum_{\route \ni e}
\flow_{\route}$. When we need multiple flow variables, we use the analogous
notation $\xvec, x_{\route}, x_e$ and $\zvec, z_{\route}, z_e$.

The network is subject to congestion, modeled with stochastic delay functions
$\ell_e(\flow_e) +\xi_e(\flow_e)$ for each edge $e\in \GA$. Here, {the
deterministic function} $\ell_e(\flow_e)$ measures the expected delay when the
edge has flow $\flow_e$, and $\xi_e(\flow_e)$ is a random variable that
represents a noise term on the delay, encoding the error that $\ell_e(\cdot)$
makes. Functions $\ell_e(\cdot)$, generally referred to as {\em latency
functions}, are assumed continuous and non-decreasing. The expected latency
along a path $p$ is given by $\ell_{\route}(f):=\sum_{e\in \route}
\ell_e(\flow_e)$.

Random variables $\xi_e(\flow_e)$ have expectation equal to zero and standard
deviation equal to $\std_e(\flow_e)$, for arbitrary continuous functions
$\std_e(\cdot)$. For the variational inequality characterization used in
Section~\ref{sec:funct}, we further assume that standard deviation functions
are non-decreasing.
%
%
We assume that these random variables are pairwise independent. From there, the
variance along a path equals $\var_{\route}(f)=\sum_{e\in \route}
\std_e^2(\flow_e)$, and the standard deviation (stdev) is
$\sigma_{\route}(f)=(\var_{\route}(f))^{1/2}$.  We will initially work with
variances and then extend the model to standard deviations, which have the
complicating square roots. (For details on the complications, we refer the
reader to \cite{nikolova-stochWE,Nikolova-Stier15}). 

We will consider the {\em nonatomic} version of the routing game where
infinitely many players control an infinitesimal amount of flow each so that
the path choice of a single player does not unilaterally affect the costs
experienced by other players (even though the joint actions of players affect
other players).

Players are risk-averse and choose paths taking into account the variability of
delays by considering a {\em mean-var} objective
$\pathcost^{\gamma}_{\route}(\flow) = \ell_{\route}(\flow)+\gamma
\var_{\route}(\flow)$.  We refer to this objective simply as the {\em path
cost} (as opposed to latency). Here, $\gamma\ge 0$ is a constant that
quantifies the risk-aversion of the players, which we assume homogeneous. The
special case of $\gamma = 0$ corresponds to risk-neutrality.

The variability of delays is usually not too large with respect to the
expected latency. It is common to consider the {\em coefficient of
{variation}} $CV_e(\flow_e):=\std_e(\flow_e)/\ell{_e}(\flow_e)$ given by the
ratio of the standard deviation to the expectation as a relative measure of
variability {\cite{McAuliffe-CV}}. In this case, we consider the {\em variance-to-mean ratio}
$\var_e(\flow_e)/\ell{_e}(\flow_e)$ as a relative measure of variability.
Consequently, we assume that $\var_e(\xx_e)/\ell_e(\xx_e)$ is bounded from
above by a fixed constant $\kappavar$ for all $e\in \GA$ at the equilibrium
flow {of interest} $\xx_e\in \R_+$,
{which is less restrictive than requiring such a bound for all feasible flows}.
This means that the variance cannot be larger than
$\kappavar$ times the expected latency in any edge at the equilibrium flow.

%
%

{In summary, an instance of the problem is given by the tuple}
$(G,d,\ell,\var,\gamma)$, which represents the topology, the demand, the
latency functions, the variability functions, and the degree of player risk-aversion.



The following definition captures that at equilibrium players route flow
along paths with minimum cost $\pathcost^{\gamma}_{\route}(\cdot)$. In
essence, users will switch routes until at equilibrium costs are equal along
all used paths.  This is the natural extension of the traditional Wardrop
Equilibrium to risk-averse users.

\begin{definition}[Equilibrium]\label{defi:eq}
A {\em $\gamma$-equilibrium} of a stochastic nonatomic routing game is a flow
$\flow$ such that for every path $p\in\paths$
with positive flow, the path cost $\pathcost^{\gamma}_{p}(\flow) \leq
\pathcost^{\gamma}_{q}(\flow)$ for any other path $q \in \paths\,$. For a
fixed risk-aversion parameter $\gamma$, we refer to a $\gamma$-equilibrium as a {\em risk-averse
Wardrop equilibrium} (RAWE), denoted by $\xx$.
\end{definition}

Notice that since the variance decomposes as a sum over all the edges that
form the path, the previous definition represents a standard Wardrop
equilibrium with respect to modified costs $\ell_e(\flow_e)+\gamma\var_e(\flow_e)$. 
For the existence of the equilibrium, it is sufficient that the modified cost functions are increasing. 

Our goal is to investigate the effect that risk-averse players have
on the quality of equilibria. The quality of a solution that represents
collective decisions can be quantified by the cost of equilibria with respect
to expected delays since, over time, different realizations of delays average
out to the mean by the law of large numbers. For this reason, a social
planner, who is concerned about the long term, is typically risk neutral, as
opposed to users who tend to be more emotional about decisions. Furthermore,
the social planner may aim to reduce long-term emissions, which would be
better captured by the total expected delay of all users. These arguments
justify the difference between the risk aversion coefficient that
characterizes user behavior at equilibrium and the behavior of the social
planner.

\begin{definition}
The {\em social cost} of a flow $\flow$ is
defined as the sum of the expected latencies of all players:
$
\CC(\flow) := \sum_{\route \in \paths} \flow_{\route}\ell_{\route}(\flow)=
\sum_{e \in \GA} \flow_e\ell_e(\flow_e)
$\,.
\end{definition}

Although one could have measured total cost as the weighted sum of the costs
$\pathcost^{\gamma}_{\route}(\flow)$ of all users, this captures users'
utilities but not the system's benefit. 
Nikolova and Stier-Moses {\cite{nikolova-stochWE} considered} such a cost function to compute
the price of anarchy; in the
current paper, our goal is to compare across different values of risk aversion
so we want the various flow costs to be measured {with the same units}.

The next definition captures the increase in social cost at equilibrium
introduced by user risk-aversion, compared to the cost one would have if users
were risk-neutral. Hence, we use a {\em risk-neutral Wardrop equilibrium} 
(RNWE), defined as a $0$-equilibrium according to Definition~\ref{defi:eq}, as
the yardstick to determine the inefficiency caused by risk-aversion. We define
the price of risk aversion as the worst-case ratio among all possible instances
of expected costs of the risk-averse and risk-neutral equilibria.

\begin{definition}[\cite{Nikolova-Stier15}]
Considering an instance family $\mathcal{F}$ of a routing game with uncertain
delays, the {\em price of risk aversion} ($\prav$) associated with $\gamma\kappavar$
(the risk-aversion coefficient times the variance-to-mean ratio) is defined by
\begin{equation}\label{eq:PoR}
{
\prav(\mathcal{F},\gamma,\kappavar):=\sup_{G,d,\ell,\var} \left\{
\qquad\frac{\CC(\xx)}{\CC(\zz)}~:~
(G,d,\ell,\var)\in \mathcal{F}\text{ and } \var(\xx)\le\kappavar\ell(\xx)
\right\},
}
\end{equation}
where $\xx$ and $\zz$ are the RAWE and the RNWE of the
corresponding instance.
\end{definition}

This supremum depends on $\mathcal{F}$, which may be defined in terms of the
network topology (as, e.g., general, series-parallel, or Braess networks), the
number of vertices, or the set of allowed latency functions (as, e.g., affine
or quadratic polynomials). Different results will be with respect to different
families $\mathcal{F}$, with Sections~\ref{sec:lb} and~\ref{sec:stdev} focusing
on topological definitions, and Section~\ref{sec:funct} focusing on sets of
allowed functions. For the sake of brevity, we will typically write just $\prav$
and the parameters $\mathcal{F}$, $\gamma$, and $\kappa$ will be implicit by
the context. Although we do not specify it explicitly in each result for
brevity, all our results work for arbitrary values of $\gamma\ge 0$ and
$\kappa\ge 0$.

%% file: structlowerbound.tex
\section{Structural Lower Bounds for the Price of Risk Aversion}\label{sec:lb}


In this section, we prove two lower bounds on the price of risk aversion
$\prav$, both matching the upper bounds presented by Nikolova and Stier-Moses
\cite{Nikolova-Stier15}. The first bound for $\prav$ is with respect to the
minimum number of alternations among all alternating paths as defined below,
while the second bound is with respect to the number of vertices in the graph.
In fact, the same bounds hold in the mean-standard deviation model, but we
defer that discussion to Section~\ref{stdevbounds}.

Given an instance, we denote a RNWE flow associated with it by $z$, and a RAWE
flow by $x$. To define alternating paths, we partition the edge-set $\GA$ into
$A=\{e\in\GA :x_e<z_e\}$ and $B=\{e\in\GA :z_e\leq x_e\}$. Conceptually, an
alternating path is an $s$-$t$-path in the graph where edges in $B$ are
reversed (see Fig.~\ref{fig:forbiddenGraphs}(c) for an illustration of the
definition). 
\begin{definition}[\cite{Nikolova-Stier15}] 
A generalized path $\pi=A_1$-$B_1$-$A_2$-$B_2$-$\cdots$-$A_t$-$B_t$-$A_{t+1}$,
composed of a sequence of subpaths, is an {\em alternating path} when every edge in
$A_i\subseteq\GA$ is directed in the direction of the path, and every edge in
$B_i\subset\GA$ is directed in the opposite direction from the path.  We say that such
a path has $t+1$ disjoint forward subpaths, and $t$ {\em alternations}.  
\end{definition}
The definition of alternating paths was motivated by the following result.

\begin{theorem}[\cite{Nikolova-Stier15}]\label{thm:ub}
Considering the set of instances with arbitrary mean and variance latency
functions that admit an alternating path with up to $\eta$ disjoint forward
subpaths, $\prav\leq 1+\eta\gamma\kappa$.
\end{theorem}

The theorem implies that for the set of instances on graphs with $n$ vertices,
$\prav\leq 1+\gamma\kappa \lceil \frac{n-1}{2} \rceil$ since in that case an
alternating path cannot have more than $\lceil (n-1)/2 \rceil$ disjoint forward
subpaths.
We are going to prove that those upper bounds are tight. To get there, we first
prove a more general result that shows how instances with high price of risk
aversion can be constructed.

\begin{theorem}\label{tightexample}
For an $i\in \N_{>0}$, consider $r^i_A,r^i_N\in \R_{\ge 0}$ such that $2^i r^i_A > (2^i-1) r^i_N$. 
There exists an instance based on a graph $G^i(r^i_A,r^i_N)$ that satisfies the following two properties.
\begin{itemize}
\item If $r^i_A$ risk-averse players are routed through $G^i(r^i_A,r^i_N)$,
then the path cost along used paths at the RAWE flow $x$, as well as
the expected latency, is $1+2^i\gamma\kappa$. The social cost is
$\CC(x)=(1+2^i\gamma\kappa)r^i_A$.
\item If $r^i_N$ risk-neutral players are routed through $G^i(r^i_A,r^i_N)$,
then the expected latency along each used path at the RNWE flow $z$
is $1$. The social cost is $\CC(z)=r^i_N$.
\end{itemize}
\end{theorem}

The proof is by induction on $i$. We will recursively construct the instance
for $i$ by forming a Braess instance with the graph resulting for the $i-1$
case. At each step we will need to find a mean latency function that makes the
properties in the statement work.

\begin{proof}
For the base case $i=1$, we let
$G^1(r^1_A,r^1_N)$ be the Braess graph, shown in Fig.~\ref{basecase}.
Indeed, consider any $r^1_A,r^1_N$ such that $r^1_A >r^1_N/2$ as indicated
in the statement of the result. We define the mean latency function $a_1(x)$ to be 
any function that is strictly increasing for $x\geq r^1_N/2$ and
such that $a_1(r^1_N/2)=0$ and $a_1(r^1_A)=\gamma\kappa$. Note that
in order for $a_1(x)$ to be strictly increasing, it is necessary that
$r^1_A>r^1_N/2$, which holds by hypothesis.

\begin{figure}[t]
\centering
\includegraphics[scale=0.56]{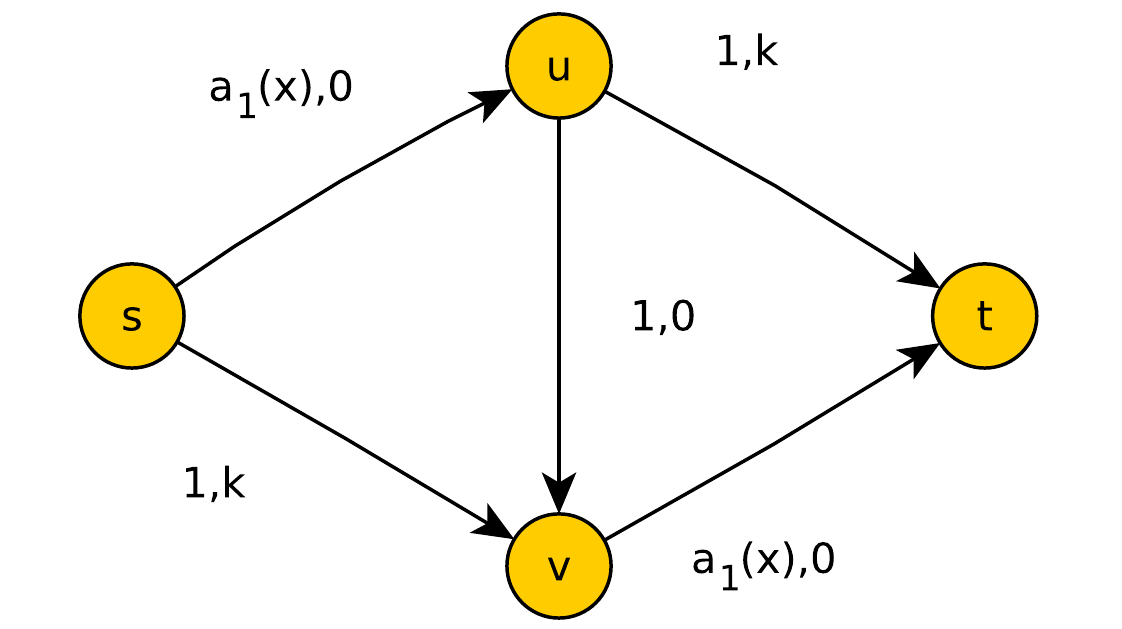}
\caption{The base case $G^1(r^1_A,r^1_N)$ is a Braess graph.}
\label{basecase}
\end{figure}

The RAWE flow $x$ routes the $r^1_A$ risk-averse players along the zig-zag
path. Hence, the mean-var objective in the upper-left and the lower-right
edges, as well as the mean latency, will each be $\gamma\kappa$, totalling
$1+2\gamma\kappa$ for each player in both cases. Hence, $\CC(x)=(1+2^i\gamma\kappa)r^i_A$.
Instead, the RNWE flow $z$ routes the $r^1_N$ risk-neutral players along the
top and bottom paths, half and half.  Hence, the cost for each player is~$1$
and $\CC(z)=r^i_N$, proving the base case.

Let us consider the inductive step where we assume we have an instance
satisfying the properties for $i-1$ and construct the instance for step $i$.
Starting from $r^i_A$ and $r^i_N$ satisfying the condition in the statement for
case $i$, we set
$r^{i-1}_A=(2^i r^i_A-r^i_N)/2^{i+1}$ and
$r^{i-1}_N=r^i_N/2$. We first verify that these values satisfy the hypothesis
for the case $i-1$. Indeed,
$r^{i-1}_A
>\frac{2^{i-1}-1}{2^{i-1}}r^{i-1}_N$ because by hypothesis $r^i_A
>\frac{2^i-1}{2^i}r^i_N\Leftrightarrow \frac{r^i_A}{2}
>\frac{2^i-1}{2^{i+1}}r^i_N$ which implies that
$\frac{r^i_A}{2}-\frac{r^i_N}{2^{i+1}}
>\frac{2^i-1}{2^{i+1}}r^i_N-\frac{1}{2^{i+1}}r^i_N
=\frac{2^{i-1}-1}{2^{i-1}}\frac{r^{i}_N}{2}$. 

Using the graph corresponding to step $i-1$ and the values of $r^{i-1}_A$ and
$r^{i-1}_N$ specified previously, we construct graph $G^i(r^{i}_A,r^{i}_N)$
with those components as shown in Fig.~\ref{inductivestep}.
We define the mean latency function $a_i(x)$ to be any function that is strictly
increasing for $x\geq r^i_N/2$ and such that $a_1(r^i_N/2)=0$ and
$a_i(\frac{r^i_A}{2}+\frac{r^i_N}{2^{i+1}})=2^{i-1}\gamma\kappa$.  Note that in order
for $a_i(x)$ to be strictly increasing, it is necessary that
$\frac{r^i_A}{2}+\frac{r^i_N}{2^{i+1}}>\frac{r^i_N}{2}$, which actually holds
because, by hypothesis, $r^i_A >\frac{2^i-1}{2^i}r^i_N\Leftrightarrow
\frac{r^i_A}{2} >\frac{2^i-1}{2^{i+1}}r^i_N$.

\begin{figure}[t]
\center
\includegraphics[scale=0.85]{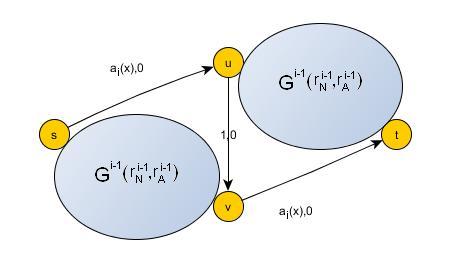}\\
\caption{The recursive construction of $G^i(r^i_A,r^i_N)$ by forming a Braess graph topology using components of the earlier step.}
\label{inductivestep}
\end{figure}

The RAWE flow $x$ routes the $r^i_A$ risk-averse players as follows:
$\frac{r^i_A}{2}-\frac{r^i_N}{2^{i+1}}$ units along the upper path,
$\frac{r^i_N}{2^{i}}$ units along the zig-zag path, and
$\frac{r^i_A}{2}-\frac{r^i_N}{2^{i+1}}$ units along the lower path. The
mean-var objective of the upper-left and the lower-right edges, as well as the
mean latency, will each be $2^{i-1}\gamma\kappa$ since the flow through them is
equal to $\frac{r^i_A}{2}+\frac{r^i_N}{2^{i+1}}$. The flow inside each of the
copies of $G^{i-1}(r^{i-1}_A,r^{i-1}_N)$ is a RAWE for which we know, by
induction, that all players perceive a path cost of $1+2^{i-1}\gamma\kappa$, which
additionally, by induction, is the mean latency of all used paths. Thus, the
path cost that players perceive in $G^i(r^i_A,r^i_N)$ under the RAWE flow $x$
is $1+2^i\gamma\kappa$, which additionally is the mean latency of all used
paths, and the social cost is $\CC(x)=(1+2^i\gamma\kappa)r^i_A$.

The RNWE flow $z$ routes the $r^i_N$ risk-neutral players along the top and
bottom paths, half and half.  Hence, the path cost for perceived by each player
is~$1$, as the mean-var objective in the upper-left and lower-right edges is
equal to 0, and, by induction, passing through either of both copies of
$G^{i-1}(r^{i-1}_A,r^{i-1}_N)$ has a mean-var objective of 1. This implies that
$\CC(z)=r^i_N$, which completes the proof.
\qed
\end{proof}

The previous result provides a constructive way to generate instances with high
price of risk aversion.  We show the concrete topology for the cases $i=2$ and
$i=3$ in Fig.~\ref{example23} below.  Notice that the paths of the instances
resulting from these constructions have at most one edge with non-zero
variance. This fact is useful to extend our lower bounds to the mean-stdev
model, since in that case summing and taking square roots is not needed.

Another useful observation is that the prevailing value for mean latency functions
$a_j$ under the RNWE flow $z$ is $0$, and under the RAWE flow $x$ is $2^{j-1}$.
This can be easily proved by induction and will be used when establishing functional
lower bounds on the $\prav$ in the next section. 

We now use the previous result to get lower bounds for $\prav$ matching the upper bound specified earlier.

\begin{corollary}\label{cor:verticestight}
For any $n_0\in\N$, there is an instance on a graph $G$ with $n\geq n_0$ vertices such that its equilibria satisfy
$\CC(\xx)\ge 1 +\gamma\kappa \lceil(n-1)/2\rceil \CC(\zz)$. 
\end{corollary}
\begin{proof}
Consider an arbitrary demand $d$, and apply Theorem~\ref{tightexample} with
$r^i_A=r^i_N=d$ and $i=\min\{j\in\N:n_0\leq 2^j\}$ to get instance
$G^i(r^i_A,r^i_N)$. Consequently, the RAWE flow $x$ and the RNWE flow $z$
satisfy that
\[\frac{\CC(x)}{\CC(z)}=\frac{(1+2^i\gamma\kappa)d}{d}=
1+\gamma\kappa \frac{n}{2}
\,,\]
because $G^i(r_A^i,r_N^i)$ has $2^{i+1}$ vertices by construction.
Finally, the result holds because $n$ is a power of two.
\qed
\end{proof}

The previous lower bound together with the upper bound given in the paragraph
after Theorem~\ref{thm:ub} imply that the $\prav$ with respect to the set of
instances on graphs with up to $n$ vertices is {\em exactly} equal to
$1+\gamma\kappa \lceil(n-1)/2\rceil$ when $n$ is a power of 2. From there, the
bound is tight infinitely often.  Although for other values of $n$ the bounds
are not tight, they are close together so these results provide an
understanding of the asymptotic growth of the $\prav$. We now refine this
observation to the bound in Theorem~\ref{thm:ub}.

\begin{theorem}\label{cor:altertight}
The upper bound for the price of risk-aversion shown in Theorem~\ref{thm:ub}
and the lower bound shown in Corollary~\ref{cor:verticestight} coincide for
graphs of size that is a power of~2. Otherwise, the gap between them is less
than~2.
\end{theorem}

\begin{proof}
For an arbitrary $i>1$, we consider the instance with $2^{i+1}=2\eta$ vertices
constructed in Corollary~\ref{cor:verticestight}. In that instance, the only
alternating path has exactly $2^i=\eta$ disjoint forward subpaths.  Indeed,
using Fig.~\ref{example23} as an example of the representation of the graph, we
define an alternating path by recursively choosing the lower component, next
the reverse vertical edge, and last recursively choosing the upper component. By
expanding both recursions, it is not hard to see that the alternating path
covers all $2\eta$ vertices, and its $\eta$ non-vertical edges are disjoint
forward subpaths, as required.  According to the equilibrium flows computed in
Corollary~\ref{cor:verticestight}, the non-vertical edges in the alternating paths
belong to $A$, while the rest of the edges belong to $B$. Hence, the
alternating path is compatible with the definitions of $A$ and $B$, as
required.

For graph sizes $n$ that are not a power of 2, there is a rounding error. For the
lower bound, we need to consider the maximum power of 2 smaller than $n$. The
relative gap satisfies 
$$\frac{UB}{LB} \le
\frac{1+\gamma\kappa\lceil (n-1)/2 \rceil}{1+\gamma\kappa 2^{\lfloor log_2(n))\rfloor}}<2\,.$$
\qed
\end{proof}

In conclusion, $\prav=1+\eta\gamma\kappa$ when the family of instances is defined as 
graphs with arbitrary mean and variance functions that admit alternating paths
with up to $\eta$ disjoint forward subpaths, for $\eta$ equal to a power of 2.
We have equality because the supremum in the definition of $\prav$ is attained
by the instance constructed previously.

%% file: functionalbounds.tex
\section{Functional Bounds}\label{sec:funct}

In this section, we turn our attention to instances with mean latency functions
restricted to be in a certain family (as, e.g., affine functions).  We prove
upper and lower bounds for the $\prav$ that are asymptotically tight as
$\gamma\kappa$ increases. The results rely on the variational inequality
approach that was first used by~\cite{css-congestion} to prove price of anarchy
(POA) bounds for fixed families of functions. This approach was based on the
properties of the allowed functions. Since then, these properties have been
successively refined by~\cite{harks-atomic,Roughgarden}, and they are now
usually referred to as the {\em local smoothness} property. Although not really
needed for the results here, we use the latter terminology since it has become
standard by now.  To characterize a family of mean latency functions, we rely on
the smoothness property, defined below.

\begin{definition}[\cite{Roughgarden}]
A function $\ell:\R_{\ge 0}\rightarrow \R_{\ge 0}$ is said to be \emph{$(\lambda,\mu)$-smooth around $x\in \R_{\ge 0}$} if 
$y \ell(x)\leq \lambda y \ell(y) +\mu x \ell(x)$ for all $y\in \R_{\ge 0}$.
\end{definition}


Using the previous definition, we construct an upper bound for the $\prav$ when
mean latency functions $\{\ell_e\}_{e\in \GA}$ are $(1,\mu)$-smooth around the
RAWE flow $x_e$ for all edges $e\in\GA$. \citet{meir-wronggame}  proved a
similar bound using a related approach in which they generalize the smoothness
definition to \emph{biased smoothness} which holds with respect to a modified
latency function. In our case, the modified latency function would be
$\ell_e+\gamma v_e$.  One advantage of our approach is its simplicity; it is a
straightforward generalization of the POA proof given in~\cite{css-congestion}.
Also, we only require smoothness around the equilibrium flows, while the biased
smoothness of \cite{meir-wronggame} requires the property for all $x\in
\R_{\geq 0}$.
We provide a proof corresponding to our assumptions, matching what is needed to
get our asymptotically-tight lower bounds.

\begin{theorem}\label{thm:functional-upper}
Consider the set of general instances with mean latency functions $\{\ell_e\}_{e\in
\GA}$ that are $(1,\mu)$-smooth around any RAWE flow $x_e$ for all $e\in\GA$.
\footnote{If the instance admits multiple equilibria, we require smoothness around all of the corresponding flows.} 
Then, with respect to that set of instances,
$$\prav\leq(1+\gamma\kappa)\frac{1}{1-\mu}.$$
\end{theorem}
\begin{proof}
We consider an instance within the family, a corresponding RAWE flow $x$, and a
RNWE flow $z$.  Further, we let $A=\{e\in\GA|x_e<z_e\}$ and
$B=\{e\in\GA|z_e\leq x_e\}$. Using a variational inequality formulation for the
RAWE \cite{nikolova-stochWE}, we have that
$$\sum_{e\in \GA}x_e(\ell_e(x_e)+\gamma v_e(x_e))  \leq  \sum_{e\in \GA}z_e(\ell_e(x_e)+\gamma v_e(x_e))\,.$$
Partitioning the sum over $\GA$ at both sides into terms for $A$ and $B$,
subtracting the following inequality
\begin{equation}\label{eqn:ineq}
\sum_{e\in A}x_e v_e(x_e) + \sum_{e\in B}x_e v_e(x_e) \geq \sum_{e\in B}z_e v_e(x_e)
\end{equation} 
from it, and further bounding $v_e(x_e)$ by $\kappa \ell_e(x_e)$, we get that
$$\CC(x)=\sum_{e\in A}x_e \ell_e(x_e)+\sum_{e\in B}x_e \ell_e(x_e) \leq \sum_{e\in A}(1+\gamma\kappa)z_e \ell_e(x_e)+\sum_{e\in B}z_e \ell_e(x_e).$$
Inequality (\ref{eqn:ineq}) follows from the non-negativeness of the flow and
the variance, and from the definition of $B$. Applying the definition of $A$ to
the first term in the right-hand side of the last inequality and the
$(1,\mu)$-smoothness condition to the second term, we upper bound the cost and
the result follows.
$$\CC(x)\le\sum_{e\in A}(1+\gamma\kappa) z_e \ell_e(z_e) + \sum_{e\in B} (z_e \ell_e(z_e) +\mu x_e \ell_e(x_e)) 
  \leq (1+\gamma\kappa)\CC(z)+\mu \CC(x).$$
\qed
\end{proof}

The bound in the previous result is similar to that for the POA for nonatomic
games with no uncertainty. Indeed, the result there is that $POA\le
(1-\mu)^{-1}\,$, the same without the $1+\gamma\kappa$ factor.  The values of
$(1-\mu)^{-1}\,$ have been computed for different families of functions in
previous work. To provide some examples, it is equal to $4/3$ for affine
latency functions, and approximately equal to 1.626, 1.896, and 2.151 for
quadratic, cubic, and quartic polynomial latency functions, respectively.  On
the other hand, for unrestricted functions this value is infinite so the bound
becomes vacuous in that case, which provides support for the topological analysis 
of Section~\ref{sec:lb}.

To evaluate the tightness of our upper bounds, we now propose lower bounds for
the $\prav$. More specifically, we provide a family of instances indexed
by $i$ whose latency functions are $(1,\mu_i)$-smooth, for $\mu_i=1-2^{-i}$,
for which the bound is approximately tight. These instances imply lower bounds
equal to $1+\gamma\kappa (1-\mu_i)^{-1}= 1+\gamma\kappa 2^i$. A few remarks
are in order. First, notice that although the lower and upper bounds do not
match, they are similar. The difference is whether the $1$ is or is not
multiplied by the $\mu$ factor. When the $\gamma\kappa$ term is large, both
bounds are essentially equal. Second, notice that for large values of
$i$, necessarily the number of alternations of the longest alternating path
must grow exponentially large to simultaneously match the structural upper
bound presented in Theorem~\ref{thm:ub}.

\begin{theorem}\label{thm:functional-lower}
For any $i>0$, letting $\mu_i=1-2^{-i}$,
$\prav\ge 1+\gamma\kappa (1-\mu_i)^{-1}\,$ for the family of instances
satisfying the $(1,\mu_i)$-smoothness property.
\end{theorem}

To get the result we use the recursive construction of
Theorem~\ref{tightexample} but with cost functions satisfying the
$(1,\mu_i)$-smoothness condition around the RAWE.  For the given $i$, we
construct a graph $G^i$ that implies that $\prav\ge 1+\gamma\kappa
2^i=1+\gamma\kappa (1-\mu_i)^{-1}\,$. A brief roadmap of the proof is as
follows. We specify the instance, and determine the RNWE flow $z$ and the RAWE
flow $x$ with their costs. From there, we conclude that $\prav\ge
1+\gamma\kappa (1-\mu_i)^{-1}$. Finally, we prove that $\ell_e$ is
$(1,\mu_i)$-smooth around $x_e$ for all $e$. 

\begin{proof}
We consider the graph $G^i$ constructed in Theorem~\ref{tightexample} with
$r^i_A=r^i_N=1$ (see Fig.~\ref{basecase} and~\ref{inductivestep}) but with
alternative functions $a_j(\cdot)$. For the functions defined at level 
$1\leq j\leq i$, we set $a_j(x)=0$ for $x\leq 2^{j-1}/2^{i}$ whereas for larger $x$, 
$a_j(x)$ increases linearly to attain $2^{j-1}\gamma\kappa$ at $x=\frac{2^{j-1}}{2^i-1}$. 
Mathematically, 
$$a_j(x)=\max\left\{ 0, 
\frac{2^{j-1}\gamma\kappa}{\frac{2^{j-1}}{2^i-1}-\frac{2^{j-1}}{2^{i}}}\Big(x-\frac{2^{j-1}}{2^{i}}\Big)\right\}\,.$$

To simplify notation, we refer to edges that have cost function $a_j(\cdot)$ as
$a_j$. Figure~\ref{example23} illustrates the construction for $i=2$ and
$i=3$. As an example, we specify the resulting mean latency functions for $i=2$:
$a_1$ is such that $a_1(\frac{1}{4})=0$ and $a_1(\frac {1}{3})=\gamma\kappa$, and
$a_2$ satisfies $a_2(\frac{1}{2})=0$ and $a_2(\frac{2}{3})=2\gamma\kappa$.  The RNWE flow  
splits the unit flow equally along the 4 paths not containing any
vertical edge. Instead, the RAWE flow splits the unit flow equally along the 3
paths that contain a vertical edge. Evaluating those functions,
$a_j(z)=0$ and $a_j(x)=2^{j-1}\gamma\kappa$, from where
$\prav\ge\CC(x)/\CC(z)=(1+4\gamma\kappa)/1=1+4\gamma\kappa$ for the family of functions that are
$(1,3/4)$-smooth.

%

\begin{figure}[t]
\center
\includegraphics[scale=0.57]{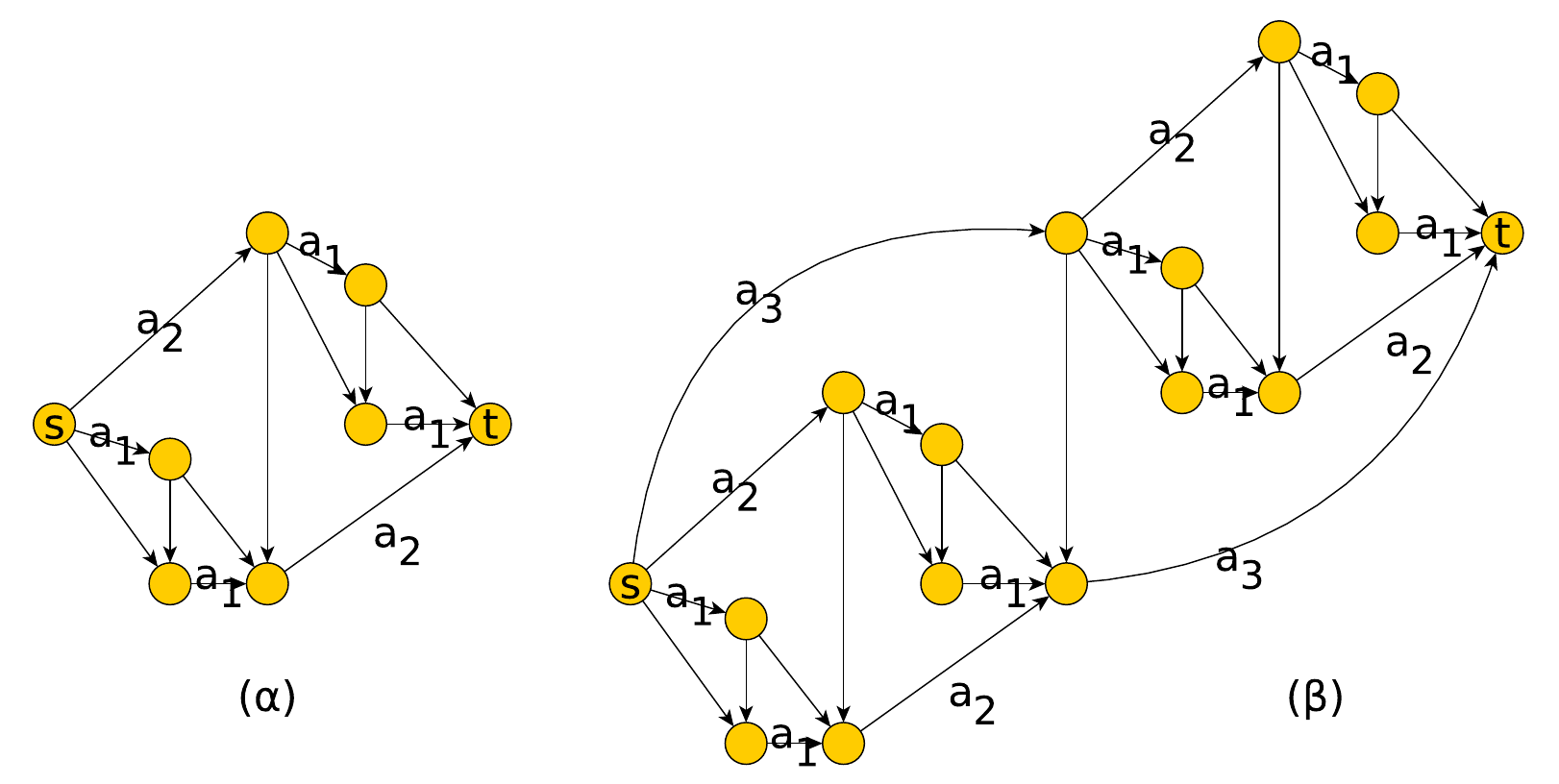}
\caption{The resulting graphs $G^i$ for $i=1$ on the left, and $i=2$ on the
right.  Edges labeled with $a_j$ have the mean latency function with equal name and
$0$ variance.  Vertical edges have mean latency functions equal to $1$ and variance equal to $0$.
Finally, the rest of the edges have mean latency functions equal to $1$ and variance
equal to $\kappa$.}\label{example23} 
\end{figure}

We refer to paths in $G_i$ not containing any vertical edge in representations
such as that of Fig.~\ref{example23} as \emph{parallel paths}. The rest of
the paths, containing a single vertical edge, are referred to as \emph{zig-zag
paths}.  It is not hard to see using an inductive proof on the construction of
$G^i$ that there are $2^i$ parallel paths and there are $2^i-1$ zig-zag paths. 

Generalizing what we saw in the example for $i=2$, the RNWE flow $z$ splits the
unit flow equally along the $2^i$ parallel paths.  To verify that $z$ is at
equilibrium, observe that for each $a_j$ edge, for $1\leq j\leq i$, there are
$2^{j-1}$ parallel paths passing through it.
Consequently, each $a_j$ will get $2^{j-1}/2^i$ units of flow, implying that
their costs are $0$.  Path costs under $z$ are thus the same as those in
Theorem~\ref{tightexample}, which implies that $z$ is indeed a RNWE and
that $\CC(z)=1$. 

On the other hand, the RAWE flow $x$ splits the unit flow equally along the
$2^i-1$ zig-zag paths.  To verify that $x$ is at equilibrium, observe that for
each $a_j$ edge, for $1\leq j\leq i$, there are $2^{j-1}$ zig-zag paths passing
through it.
Consequently, each $a_j$ will get $2^{j-1}/(2^i-1)$ units of flow, implying that
their costs are $2^{j-1}\gamma\kappa$. Path costs under $x$ are thus the same as those in
Theorem~\ref{tightexample}, which implies that $x$ is indeed the RAWE and
that $\CC(x)=1+2^i\gamma\kappa$. From there, the bound for $\prav$ in the statement of
the theorem follows.

What remains to be shown is that for any chosen $G^i$, functions $a_j$ are
$(1,\mu_i)$-smooth around $x_j$, where $x_j=2^{j-1}/(2^i-1)$ is the RAWE flow
at edge $a_j$. The other cost functions are constant so they trivially satisfy
the smoothness properties.  To prove this, let us consider $1\leq j\leq i$.
From the definition of smoothness we need to show that $y a_j(x_j)\leq y
a_j(y)+\mu_i x_j a_j(x_j)$ for all $y\in\R_{\ge 0}$.  Equivalently, we can show
that
$$1-2^{-i}=\mu_i\geq \frac{\max_{y\in\R_{\ge 0}}y(a_j(x_j)-a_j(y))}{x_j a_j(x_j)}\,.$$
First, note that the maximum is attained in the interval $[2^{j-1}/2^{i},x_j]$
since $a_j(y)=0$ to the left of the interval, and the argument of the maximum
becomes negative to the right. Since $a_j$ is linear in that interval, we solve
the maximum problem by extending it linearly to the whole domain.  Since
additive constants are irrelevant because we maximize the difference of $a_j$
evaluated in two points, we modify the linearized function and add a constant
so it evaluates to 0 at 0.  For linear functions that cross the origin, the
maximizer of the problem is $x_j/2$ \cite{css-congestion}.  Because $x_j/2$ is
to the left of the interval where the maximizer must be, the maximizer with respect to $a_j$ is
$y^*=2^{j-1}/2^{i}$ and $a_j(y^*)=0$.
From there,
$$\frac{\max_{y\in\R_{\ge 0}}y(a_j(x_j)-a_j(y))}{x_j a_j(x_j)}
=\frac{y^*}{x_j}=\frac{\frac{2^{j-1}}{2^{i}}}{\frac{2^{j-1}}{2^i-1}}=\frac{2^i-1}{2^i}=\mu_i\,.$$
\qed
\end{proof}

%% file: structupperbound.tex
\section{The Mean-Standard Deviation Model}\label{sec:stdev}
\label{stdevbounds}

In this section, we turn to the mean-standard deviation model and prove upper
and lower structural bounds on $\pras$, now assuming that $\kappa$ is the
maximum among edges of the coefficient of variation $CV_e(f_e)$, defined as the
ratio between the standard deviation and the mean.  The lower bounds follow
from the same instances that were used to prove the lower bounds in the
mean-variance case in Section~\ref{sec:lb}. Considering families of graphs with
up to $\eta$ forward disjoint subpaths and general mean latency and standard
deviation functions, we prove upper bounds for the cases $\eta=\{1,2\}$, and
lower bounds for arbitrary $\eta$. For $\eta\le 2$, both bounds are
valid and coincide, so this analysis characterizes the $\pras$ for the standard
deviation case exactly.

Since we are now dealing with standard deviations, we redefine
$\pathcost^{\gamma}_{\route}(\flow) = \ell_{\route}(\flow)+\gamma
\std_{\route}(\flow)$, where $\std_{\route}(\flow)= (\sum_{e\in\route}
\std^2_e(\flow_e))^{1/2}$. Given an instance based on a graph $G$, we refer to
a RNWE flow by $z$ and to a RAWE flow by $y$. Further, we denote the (constant)
path cost perceived by players at the RAWE $y$ by $\pcdev(y)$ and the
(constant) expected latencies perceived by players at the RNWE $z$ by $\pc(z)$.
The definition for the price of risk aversion is analogous to that given in
Section~\ref{sec:model}.

Our first result provides inequalities that relate the social cost at
equilibrium with the perceived utilities. The first part is known and the
second is an easy generalization.

\begin{proposition}\label{props:eqRelSC}
For an arbitrary instance, $\CC(z)= d\pc(z)$ and $\CC(y)\leq d\pcdev(y)$, where $d$ is the traffic demand. 
\end{proposition}
\begin{proof}
Using that at equilibria all used paths have equal path costs, we get:
\begin{enumerate}
\item 
$\CC(z)=\sum_{p\in\paths} z_p\ell_p(z)= d\pc(z)$
\item 
$\CC(y)=\sum_{p\in\paths} y_p\ell_p(y)\leq \sum_{p\in\paths} y_p(\ell_p(y)+\gamma\sigma_p(y))= d\pcdev(y)$.\hfill\qed
\end{enumerate}
\end{proof}

As in Section~\ref{sec:lb}, we partition the edge-set $\GA$ into two:
$A=\{e\in\GA:y_e<z_e\}$ and $B=\{e\in\GA:z_e\leq y_e\}$.  We assume that all
edges in $\GA$ are used either by either flow $y$ or $z$.  This is without loss
of generality because unused edges can be deleted without any consequence.
The definition of $A$ implies that $z_e>0$ for all $e\in A$, while 
the assumption implies that $y_e>0$ for all $e\in B$.
To prove an upper bound on $\pras$, we rely again on alternating paths.
By Lemma~4.5 of \cite{Nikolova-Stier14wp}, we know that an alternating path
must always exist. The next result specifies that when the alternating path is
simple (i.e., it is an actual $s$-$t$ path), then the $\prav$ is low.  This result,
which is related to the result for series-parallel graphs in \cite{Nikolova-Stier15},
sets the stage for the more ambitious result below.


\begin{lemma}\label{lem:zeroAlt}
Considering the set of instances on general topologies with arbitrary mean latency
and standard deviation functions that admit alternating paths that are (actual)
paths (i.e., $0$ alternations), $\prav\leq 1+\gamma \kappa$.
\end{lemma}
\begin{proof}
We let $\pi$ be an alternating path consisting of just arcs in $A$ (i.e., it is
an actual $s$-$t$ path).  Let $p$ be any used path under the RAWE $y$. Using
the equilibrium conditions, the relationship between 1-norms and 2-norms, and
the coefficient of variation bound $\kappa$, 
$$\pcdev(y)\le\sum_{e\in p}\ell_e(y)+\gamma\sqrt{\sum_{e\in p}\sigma^2_e(y)}
\leq  \sum_{e\in \pi}\ell_e(y)+\gamma\sqrt{\sum_{e\in \pi}\sigma^2_e(y)}
\leq (1+\gamma\kappa) \sum_{e\in \pi}\ell_e(y).$$
Using the monotonicity of the mean latency functions, $\ell_e(y)\leq\ell_e(z)$ for $e\in\pi$, from where
$\pcdev(y)\leq (1+\gamma\kappa) \sum_{e\in \pi}\ell_e(z)$.
Since $z_e>0$ for all $e\in\pi$ because of the property stated earlier, a flow decomposition can be found where 
$\pi$ is used under the RNWE $z$, implying that $\CC(z)=d \sum_{e\in \pi}\ell_e(z)$.
Finally, using Proposition \ref{props:eqRelSC}, $\CC(y)\leq
d\pcdev(y)\le d (1+\gamma\kappa) \sum_{e\in \pi}\ell_e(z)= (1+\gamma\kappa)\CC(z)$,
from where we get the bound on $\pras$.
\qed
\end{proof}

We now generalize the result for Braess graphs given in \cite{Nikolova-Stier15}
to general graphs that admit alternating paths with a single alternation.

\begin{lemma}\label{lem:oneAlt}
Considering the set of instances on general topologies with arbitrary mean latency
and standard deviation functions that admit alternating paths with $2$ disjoint
forward subpaths, $\prav\leq 1+2\gamma \kappa$.
\end{lemma}
\begin{proof}
%
\begin{figure}[t]
\centering
\includegraphics[scale=0.5]{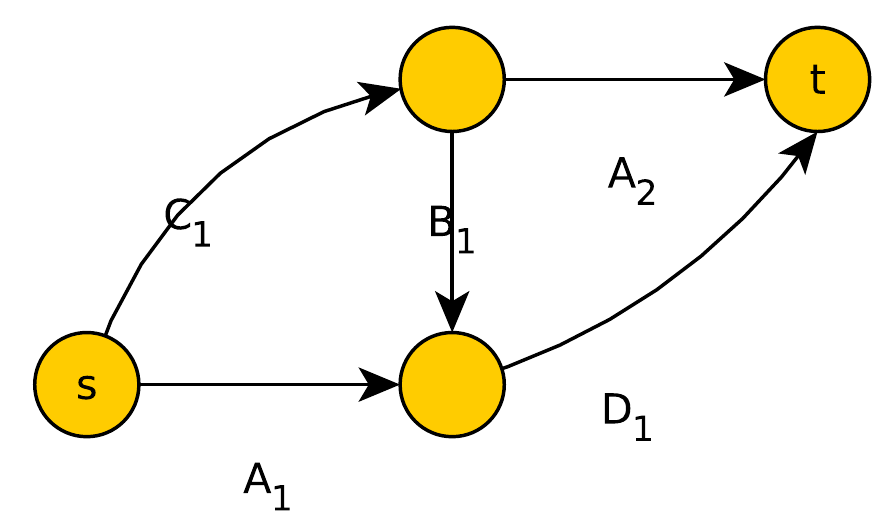}
\caption{$A_1$-$B_1$-$A_2$ is an alternating path with 2 disjoint forward subpaths.}
\label{alternatingfromscratch}
\end{figure}
%
We let $\pi=A_1$-$B_1$-$A_2$ be the alternating path with 2 disjoint forward
subpaths $A_1$ and $A_2$, and reverse subpath $B_1$, where these subpaths belong
to $A$ or $B$ correspondingly.  Figure~\ref{alternatingfromscratch} illustrates
the topology of these subpaths.  Consider a RAWE flow $y$ and a flow-carrying
path $C_1$-$B_1$-$D_1$ under $y$.  Such a path must exist because edges in
$B_1$ carry flow under $y$ as it was mentioned earlier. Using the equilibrium
conditions for $y$, 
\begin{equation}\label{ineq:eqcondy}
\ell_{C_1}(y)+\ell_{B_1}(y)+\ell_{D_1}(y)+\gamma\sqrt{\sigma^2_{C_1}(y)+\sigma^2_{B_1}(y)+\sigma^2_{D_1}(y)}\leq
\ell_{A_1}(y)+\ell_{D_1}(y)+\gamma\sqrt{\sigma^2_{A_1}(y)+\sigma^2_{D_1}(y)}. 
\end{equation}
Let us first assume that $\sigma^2_{C_1}(y)+\sigma^2_{B_1}(y)+\sigma^2_{D_1}(y)\leq \sigma^2_{A_1}(y)+\sigma^2_{D_1}(y)$.
For $\alpha\le\beta$ and $\delta\ge0$, it can be proved that $\sqrt{\beta+\delta}-\sqrt{\alpha+\delta} \leq \sqrt{\beta}-\sqrt{\alpha}$.
Letting $\alpha=\sigma^2_{C_1}(y)+\sigma^2_{B_1}(y)$, $\beta=\sigma^2_{A_1}(y)$, and $\delta=\sigma^2_{D_1}(y)$, 
(\ref{ineq:eqcondy}) implies that 
\begin{equation}
\label{ineq:pathcostC1B1}
\ell_{C_1}(y)+\ell_{B_1}(y)+\gamma\sqrt{\sigma^2_{C_1}(y)+\sigma^2_{B_1}(y)}\leq
\ell_{A_1}(y)+\gamma\sqrt{\sigma^2_{A_1}(y)}\leq (1+\gamma\kappa)\ell_{A_1}(y). 
\end{equation}
If $\sigma^2_{C_1}(y)+\sigma^2_{B_1}(y)+\sigma^2_{D_1}(y)> \sigma^2_{A_1}(y)+\sigma^2_{D_1}(y)$,
(\ref{ineq:eqcondy}) implies that 
$\ell_{C_1}(y)+\ell_{B_1}(y)\leq \ell_{A_1}(y)$, from where
$$ \ell_{C_1}(y)+\ell_{B_1}(y)+\gamma\sqrt{\sigma^2_{C_1}(y)+\sigma^2_{B_1}(y)}\leq (1+\gamma\kappa)( \ell_{C_1}(y)+\ell_{B_1}(y)) \leq (1+\gamma\kappa)\ell_{A_1}(y).$$
Hence, in both cases the same inequality holds.
%
%
%
Using a similar argument for the path $C_1$-$A_2$ instead of $A_1$-$D_1$, we get
%
$\ell_{D_1}(y)+\ell_{B_1}(y)+\gamma(\sigma^2_{D_1}(y)+\sigma^2_{B_1}(y))^{1/2}\leq (1+\gamma\kappa)\ell_{A_2}(y)$.
From the monotonicity of the square root, we further derive that 
\begin{equation}
\label{ineq:pathcostB1D1}
\ell_{D_1}(y)+\gamma\sigma_{D_1}(y)\leq \ell_{D_1}(y)+\gamma\sqrt{\sigma^2_{D_1}(y)+\sigma^2_{B_1}(y)}\leq (1+\gamma\kappa)\ell_{A_2}(y)-\ell_{B_1}(y). 
\end{equation}
Since the path is used under $y$,
$\pcdev(y)=\ell_{C_1}(y)+\ell_{B_1}(y)+\ell_{D_1}(y)+\gamma(\sigma^2_{C_1}(y)+\sigma^2_{B_1}(y)+\sigma^2_{D_1}(y))^{1/2}$.
Using again the norm-1, norm-2 inequality, together with (\ref{ineq:pathcostC1B1}) and (\ref{ineq:pathcostB1D1}), and the 
definitions of $A$ and $B$, we can upper bound the
previous expression with 
\begin{multline}
 \label{ineq:equily}
\ell_{C_1}(y)+\ell_{B_1}(y)+\ell_{D_1}(y)+\gamma\sqrt{\sigma^2_{C_1}(y)+\sigma^2_{B_1}(y)}+\gamma\sigma_{D_1}(y) 
\leq (1+\gamma\kappa)(\ell_{A_1}(y)+\ell_{A_2}(y))-\ell_{B_1}(y)\\
\leq (1+\gamma\kappa)(\ell_{A_1}(z)+\ell_{A_2}(z))-\ell_{B_1}(z). 
\end{multline}

Now we derive a related bound for the RNWE flow $z$.  For a potentially
different $D_1$, the path $A_1$-$D_1$ must carry flow under $z$.  Such a path must
exist because $z_e>0$ for $e\in A_1$. Furthermore, by the equilibrium
conditions for $z$ applied to path $A_2$, which also carries flow under $z$,
$\ell_{A_2}(z)\leq \ell_{B_1}(z)+\ell_{D_1}(z)$. Putting both remarks together,
$\pc(z)=\ell_{A_1}(z)+\ell_{D_1}(z)\geq \ell_{A_1}(z)+ \ell_{A_2}(z)- \ell_{B_1}(z)$.
Combining the last inequality with (\ref{ineq:equily}), and the fact that 
$\pc(z)$ is an upper bound for both $\ell_{A_1}(z)$ and $\ell_{A_2}(z)$, we get 
$$\pcdev(y)\leq \pc(z) + \gamma\kappa (\ell_{A_1}(z)+ \ell_{A_2}(z))\leq (1+2\gamma\kappa )\pc(z).$$
That implies the result since Proposition~\ref{props:eqRelSC} yields
$\CC(y)\leq d\pcdev(y)\leq d (1+2\gamma\kappa )\pc(z)=(1+2\gamma\kappa) \CC(z)$.
\qed
\end{proof}

%
%

Lemma~\ref{lem:zeroAlt} implies that series-parallel graphs satisfy that
$\pras\leq 1+\gamma\kappa$ \cite{Nikolova-Stier15} and it is well know that
those graphs can be characterized as not containing the Braess graph.  Along
that line, we provide a forbidden minor characterization of
Lemma~\ref{lem:oneAlt}: the topology of instances for which $\pras\leq
1+2\gamma\kappa$ is characterized by graphs not containing the domino-with-ears
graph as a minor. This graph is defined as a domino graph
\cite{graphclasses} with two arcs joining the respective pair of vertices at distance 2.
The result extends one in \cite{Nikolova-Stier15}
that establishes that $\pras\leq 1+2\gamma\kappa$ for the Braess graph.
Figure~\ref{fig:forbiddenGraphs} depicts the forbidden graphs associated with
the previous two lemmas.  As a quick checkup of our results, note that the
domino-with-ears graph contains the Braess graph as a minor so it cannot belong
to the family of instances for which $\pras\leq 1+\gamma\kappa$.  Indeed,
contracting the lower-left and upper-right vertical edges in the
domino-with-ears graph produces the Braess graph.  

\begin{figure}[t]
\center
\includegraphics[scale=0.55]{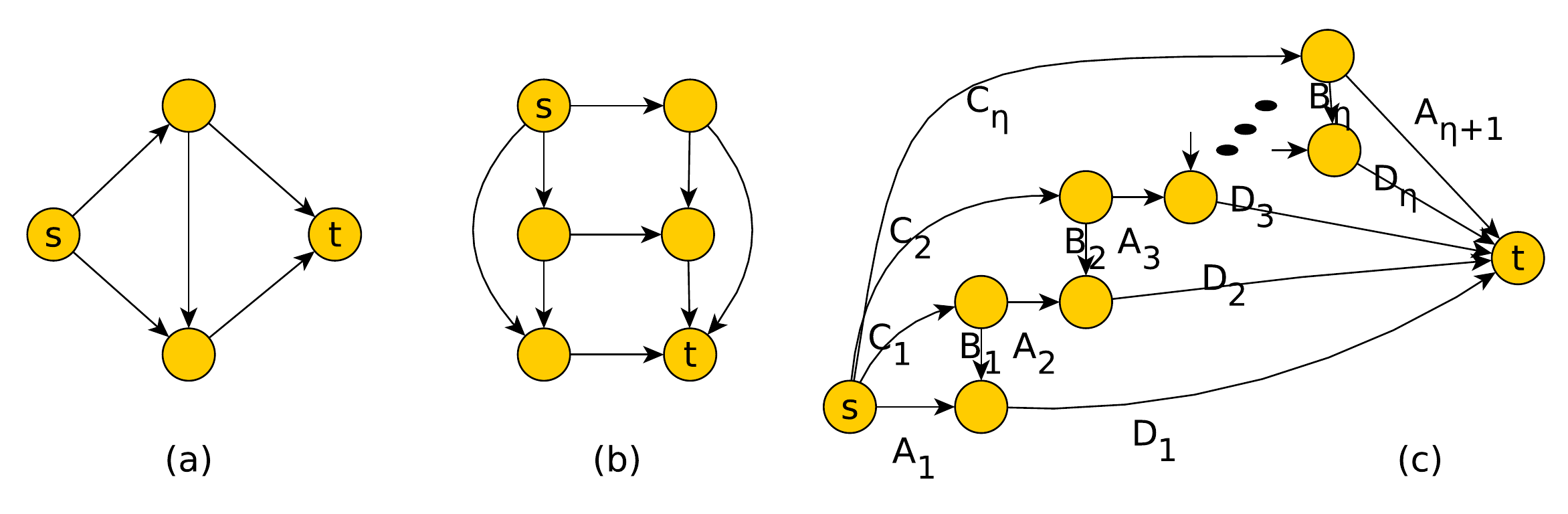}
\caption{(a) The Braess graph. (b) The domino-with-ears graph.
(c) A graph admitting an
alternating path with $\eta+1$ disjoint forward subpaths.
}
\label{fig:forbiddenGraphs}
\end{figure}

\begin{theorem}\label{thm:ear}
Considering the set of instances on general topologies with arbitrary mean latency
and standard deviation functions that do not have domino-with-ears as a minor,
$\prav\leq 1+2\gamma \kappa$.
\end{theorem}
\begin{proof}
To reach a contradiction, assume the equilibria of an instance not containing a
domino-with-ears graph as a minor satisfy that $\CC(x)>\CC(y)(1+2\gamma
\kappa)$. Since an alternating path must exist \cite{Nikolova-Stier15}, by
Lemmas~\ref{lem:zeroAlt} and~\ref{lem:oneAlt}, it must have strictly more than 2
disjoint forward subpaths.  Let us denote the alternating path by
$A_1$-$B_1$-$A_2$-$\ldots$-$B_\eta$-$A_{\eta+1}$ as indicated in
Fig.~\ref{fig:forbiddenGraphs}(c), for $\eta\geq 2$. For each $i$, a path
$C_i$-$B_i$-$D_i$ must exist because $B_i$ is a flow-carrying subpath under the
flow $y$.  By deleting all vertices not belonging to $A_i$, $B_i$, $C_i$ or
$D_i$, deleting all edges in $C_3,\ldots, C_\eta$ and $D_4,\ldots, D_\eta$,
contracting the subpath $D_3$ so its tail and head become a single vertex, and
contracting all edges of each of the subpaths
$C_1,C_2,D_1,D_2,A_1,A_2,A_3,B_1,B_2$ into a single edge, we obtain the
domino-with-ears graph. This is a contradiction to the assumption that the
domino-with-ears graph was not a minor.
\qed\end{proof}

The previous proof implies that any graph that has the 
domino-with-ears graph as a minor cannot belong to the family of
instances for which $\pras\leq 1+2\gamma\kappa$.  
Lemma~\ref{lem:zeroAlt} and Theorem~\ref{thm:ear} together imply that $\pras\leq
1+\eta \gamma\kappa$ holds for the set of instances admitting alternating paths
in which the number of disjoint forward subpaths is not more than $\eta$, for
$\eta\le 2$. Although this does not fully generalize Theorem~\ref{thm:ub} to
the case of standard deviations, it a first step in that direction.

To provide matching lower bounds for the results of this section, we note that
all paths in the proof of Theorem~\ref{tightexample} have at most one edge with
nonzero variance. For this reason, all the lower bounds in Section~\ref{sec:lb}
work by reinterpreting the variances as standard deviations and the
variance-to-mean ratios as coefficients of variation.  In summary, we can copy
the lower bound elements of Corollaries~\ref{cor:verticestight} and Theorem~\ref{cor:altertight}
to have the results for the standard deviation case.

\begin{corollary}\label{cor:altertightSD}
The upper bounds $\pras\le 1+\gamma\kappa$ and $\pras\le 1+2\gamma\kappa$
corresponding to graphs that admit alternating paths with 1 and 2 disjoint
forward subpaths, respectively, are tight.
\end{corollary}

%% file: conclusion.tex
\section{Conclusion}

We have considered the effect of risk averse players on selfish routing with
stochastic travel times, captured by mean and variance functions of flow,
following the mean-var and mean-stdev risk models in Nikolova and
Stier-Moses~\cite{Nikolova-Stier15}.  Our main conceptual contribution is a new
perspective and understanding of efficiency loss due to risk-averse behavior in
terms of {\em topological} versus {\em functional} measures, the first one
depending on the topology of the network and independent of the expected
latency functions, and the second depending on the class of allowed latency
functions and independent of the network topology, similarly to previous price
of anarchy analysis.  Our main technical contribution is the inductive
construction of a family of graphs that can be adapted with appropriate mean
and variance functions to yield both topological and functional lower bounds.
We also show how to generalize previous price of anarchy analysis for
deterministic congestion games based on variational
inequalities~\cite{css-congestion} to provide a functional upper bound here.
Our results may in turn inspire a re-investigation of the classic price of
anarchy results in deterministic settings through the lens of the topological
analysis.  

We are just at the start of understanding and characterizing the effect of
risk-averse player preferences on network equilibria.  This work opens the way
to multiple new research directions, including investigating alternative risk
models, heterogeneous players and multiple demands.  A key challenge is to
develop a better understanding and technical approaches to non-additive
risk-models, such as the mean-stdev model, which have so far resisted fully
general upper bound analysis for arbitrary graphs---a step in that direction is
our third technical contribution on the mean-stdev model for a family of graphs
that contain and generalize series-parallel graphs and the Braess graphs.  
Moreover, although the construction used for the various lower bounds 
can also provide functional lower bounds for the standard deviation case,
functional upper bounds for the standard deviation case so far remain
elusive and will be the subject of future research.

\subsubsection{ACKNOWLEDGEMENTS.\\} This work was supported
in part by NSF grant numbers CCF-1216103, CCF-1350823 and CCF-1331863, a Google Faculty Research Award, CONICET Argentina Grant PIP
112-201201-00450CO and ANPCyT Argentina Grant PICT-2012-1324. We would also like to thank the participants at the Dagstuhl Seminar on Dynamic Traffic Models in Transportation Science.